\documentclass[conference]{IEEEtran}
\IEEEoverridecommandlockouts
\usepackage[hyphens]{url}
\usepackage[bookmarks=false]{hyperref}
\usepackage{cite}
\usepackage{array}
\usepackage{amsmath,amssymb,amsfonts}
\usepackage{algorithmic}
\usepackage{graphicx}
\ifCLASSOPTIONcompsoc
	\usepackage[caption=false, font=normalsize, labelfont=sf, textfont=sf]{subfig}
\else
	\usepackage[caption=false, font=footnotesize]{subfig}
\fi
\usepackage{textcomp}
\usepackage{enumitem}
\usepackage{xcolor}

\usepackage[ruled,linesnumbered,vlined]{algorithm2e}
\usepackage{amsthm}
\usepackage{mleftright}
\mleftright
\usepackage[binary-units]{siunitx}
\DeclareSIUnit{\belmilliwatt}{Bm}
\DeclareSIUnit{\dBm}{\deci\belmilliwatt}
\DeclareSIUnit{\bps}{bps}

\usepackage[T1]{fontenc} 
\usepackage[cmintegrals]{newtxmath}
\usepackage{bm} 

\renewcommand{\qedsymbol}{$\blacksquare$}

\newcommand{\minimize}{\mathop{\hbox{minimize}}}
\newcommand{\maximize}{\mathop{\hbox{maximize}}}

\newcommand{\argmax}{\mathop{\hbox{argmax}}}
\newcommand{\subjto}{\mathop{\hbox{subject to}}}

\newtheorem{theorem}{Theorem}

\theoremstyle{definition}
\newtheorem{definition}{Definition}
\theoremstyle{remark}

\newcommand{\bp}{\ensuremath{{\mathbf p}}}
\newcommand{\bq}{\ensuremath{{\mathbf q}}}

\newcommand{\blambda}{\ensuremath{{\bm \lambdaup}}}
\newcommand{\bv}{\ensuremath{{\mathbf v}}}

\newcommand{\barbp}{\ensuremath{{\bar \bp}}}
\newcommand{\barbq}{\ensuremath{{\bar \bq}}}

\newcommand{\calN}{\ensuremath{{\mathcal N}}}
\newcommand{\calP}{\ensuremath{{\mathcal P}}}
\newcommand{\calQ}{\ensuremath{{\mathcal Q}}}
\newcommand{\calS}{\ensuremath{{\mathcal S}}}

\newcommand{\Pmax}{\ensuremath{{P_{\textnormal{max}}}}}
\newcommand{\Problem}{\ensuremath{{\textnormal{(P)}}}}
\newcommand{\Dual}{\ensuremath{{\textnormal{(D)}}}}
\newcommand{\Dualh}{\ensuremath{{\textnormal{(D$^{\bh}$)}}}}
\newcommand{\Qone}{\ensuremath{{\textnormal{(Q$_1$)}}}}
\newcommand{\Qtwo}{\ensuremath{{\textnormal{(Q$_2$)}}}}
\newcommand{\Qthree}{\ensuremath{{\textnormal{(Q$_3$)}}}}
\newcommand{\Qfour}{\ensuremath{{\textnormal{(Q$_4$)}}}}
\newcommand{\Qfive}{\ensuremath{{\textnormal{(Q$_5$)}}}}
\newcommand{\Ptwo}{\ensuremath{{\textnormal{(P$_2$)}}}}
\newcommand{\Ptwop}{\ensuremath{{\textnormal{(P$_2'$)}}}}

\newcommand{\bh}{\ensuremath{{\mathbf h}}}

\newcommand{\bR}{\ensuremath{{\mathbf R}}}

\newcommand{\E}{\ensuremath{{\mathbb E}}}

\begin{document}

\title{Low-Complexity Joint User and Power Scheduling for Downlink NOMA over Fading Channels\\
}



\author{\IEEEauthorblockN{Do-Yup Kim\IEEEauthorrefmark{1}, Hamid Jafarkhani\IEEEauthorrefmark{2}, Jang-Won Lee\IEEEauthorrefmark{1}}
\IEEEauthorblockA{\IEEEauthorrefmark{1}Department of Electrical and Electronic Engineering, Yonsei University, Seoul, Korea}
\IEEEauthorblockA{\IEEEauthorrefmark{2}Center for Pervasive Communications and Computing, University of California, Irvine, CA, USA}
\IEEEauthorblockA{Email: danny.doyup.kim@yonsei.ac.kr, hamidj@uci.edu, jangwon@yonsei.ac.kr}}

\maketitle

\begin{abstract}
Non-orthogonal multiple access (NOMA) has been considered one of the most promising radio access techniques for next-generation cellular networks.
In this paper, we study the joint user and power scheduling for downlink NOMA over fading channels.
Specifically, we focus on a stochastic optimization problem to maximize the weighted average sum rate while ensuring given minimum average data rates of users.
To address this problem, we first develop an opportunistic user and power scheduling algorithm (OUPS) based on the duality and stochastic optimization theory.
By OUPS, the stochastic problem is transformed into a series of deterministic ones for the instantaneous weighted sum rate maximization for each slot.
Thus, we additionally develop a heuristic algorithm with very low computational complexity, called user selection and power allocation algorithm (USPA), for the instantaneous weighted sum rate maximization problem.
Via simulation results, we demonstrate that USPA provides near-optimal performance with very low computational complexity, and OUPS well guarantees given minimum average data rates.
\end{abstract}

\begin{IEEEkeywords}
Fading channels, non-orthogonal multiple access (NOMA), opportunistic scheduling, power allocation, user scheduling.
\end{IEEEkeywords}

\section{Introduction}


Non-orthogonal multiple access (NOMA) is envisioned as a promising technology for future cellular networks due to its advantage of achieving high spectral efficiency over the conventional orthogonal multiple access (OMA) techniques~\cite{ding2017survey}.
Although various NOMA techniques have been proposed thus far, this paper focuses on power-domain NOMA~\cite{islam2017power}.
In power-domain NOMA, superposition coding (SPC) and successive interference cancellation (SIC) are employed at the transmitter and receiver sides, respectively.
Based on SPC and SIC, power domain diversity gain can be achieved by appropriate user selection and power allocation, which are therefore very important research topics.


There are many studies on user selection and/or power allocation in NOMA systems~\cite{fu2019joint, nain2018low, di2016sub, seo2019high, yang2017optimality, cui2018optimal}.
Specifically, in~\cite{fu2019joint}, a heuristic algorithm for joint user selection and power allocation to maximize the sum rate has been developed.
On the other hand, the weighted sum rate maximization has been studied in~\cite{nain2018low, di2016sub}.
In~\cite{nain2018low}, the authors have developed a heuristic user selection and power allocation algorithm that iteratively decides whether to select one or both for every two users.
In~\cite{di2016sub}, the authors have developed an optimal power allocation algorithm based on geometric programming (GP) and a user selection algorithm.
Then, they have alternately updates user selection and power allocation based on the matching game.
In~\cite{seo2019high}, a NOMA-based mixture transceiver architecture that applies SPC and SIC to each group and opportunistically increases the multiplexing gain while providing full diversity order is proposed.
Despite the achievement of high throughput in \cite{fu2019joint, nain2018low, di2016sub, seo2019high}, quality of service (QoS) constraints have not been considered.


In~\cite{yang2017optimality, cui2018optimal}, minimum data rate requirements are considered as the QoS constraints.
Specifically, in~\cite{yang2017optimality}, the authors have proven that the power allocation problem for the sum rate maximization is a convex problem that can be solved by standard algorithms for convex optimization.
In~\cite{cui2018optimal}, the authors have developed an optimal user selection and power allocation algorithm for the sum rate maximization based on the exhaustive search and branch-and-bound approaches.
However, due to its prohibitive computational complexity, they have proposed another heuristic algorithm using matching theory and successive convex approximation.

Despite extensive studies on user selection and/or power allocation in NOMA systems, there are still some limitations.
First, most of the previous studies, including~\cite{fu2019joint, nain2018low, di2016sub, yang2017optimality, cui2018optimal, seo2019high}, have focused on optimization problems from a snapshot perspective.
That is, user selection has been concentrated, rather than user scheduling, under fixed channel conditions while considering instantaneous QoS constraints.
In practical systems over time-varying fading channels, a significant drop in overall performance can occur since such instantaneous QoS constraints should be always satisfied even for users with very low channel gains.
However, such instantaneous QoS constraints are not essential in most applications.
In addition, although various problems have been studied, including~\cite{fu2019joint, nain2018low, di2016sub, yang2017optimality, cui2018optimal, seo2019high}, the weighted sum rate maximization problem with QoS constraints, which is NP-hard, has not been well studied.
In particular, most of the algorithms for maximizing the weighted sum rate have very high computational complexity.

Motivated by the above observations, in this paper, we develop a joint user and power scheduling algorithm with very low computational complexity in downlink NOMA over fading channels to maximize the weighted average sum rate while ensuring given minimum average data rates of users.
To this end, we first develop an opportunistic user and power scheduling algorithm (OUPS) based on the duality and stochastic optimization theory.
As a merit, OUPS is an online algorithm that can make decisions with only the instantaneous channel status without requiring knowledge of the underlying distribution of the fading channels, making it effectively applicable to various practical applications.
However, OUPS necessitates solving the user selection and power allocation problem for maximizing the instantaneous weighted sum rate at every slot.
Although its optimal solution can be achieved by the algorithm developed in~\cite{di2016sub}, the algorithm is too complex to be executed at every slot.
Hence, we additionally develop a very simple heuristic algorithm, called user selection and power allocation algorithm (USPA), which has extremely low computational complexity.
Finally, through the simulation results, we verify that USPA provides near-optimal performance notwithstanding very low computational complexity, and show that OUPS with USPA well guarantees given minimum average data rates of all users.


The rest of the paper is organized as follows.
In Section~\ref{sec:sys}, we present the system model.
In Section~\ref{sec:OUPS}, we formulate a stochastic user and power scheduling problem and develop OUPS to solve it.
In Section~\ref{sec:USPA}, we develop USPA with very low computational complexity, which is exploited at every slot in OUPS.
Simulation results are presented in Section~\ref{sec:sim}, followed by the conclusion in Section~\ref{sec:conc}.


\subsubsection*{Notation}
Scalars, vectors, and sets are denoted by italic, boldface, and calligraphic letters, respectively.
$\E[\cdot]$ denotes the statistical expectation operator.
$\mathbb{R}_{\ge0}$ denotes the set of nonnegative real numbers.
For a set~$\calS$, $(a_i)_{\forall i\in\calS}$ denotes a vector that consists of elements in the set $\{a_i : i\in\calS\}$.
For a complex number~$a$, $\lvert a \rvert$ denotes its absolute value.


\section{System Model}
\label{sec:sys}
We focus on the downlink of a single-cell in the NOMA system, where one single-antenna base station (BS) with maximum transmission power of~$\Pmax$ transmits data to a set of $N$ single-antenna users, denoted by $\calN=\{1, 2, \ldots, N\}$.
We assume a time-slotted system over block fading channels, in which channel gains vary from one slot to another but remain constant during a slot.
Let $\{h_i^t, t=1,2,\ldots\}$ be the fading process associated with User~$i$, where $h_i^t$ is a complex-valued continuous random variable representing the channel gain from the BS to User~$i$ in slot~$t$.
The fading process is assumed to be stationary and ergodic.
We assume that information on the underlying distributions of the fading process is unknown to the BS due to the difficulty of obtaining such information a priori.
However, we assume that instantaneous channel gains are known to the BS at the beginning of each slot.

In NOMA, multiple users can be simultaneously scheduled with different levels of power in the same slot.
Let $x_i^t$, satisfying $\E[\lvert{x_i^t}\rvert^2]=1$, be the information-bearing signal transmitted to User~$i$ in slot~$t$, and $p_i^t$ be the power allocated to signal $x_i^t$.
Also, let $q_i^t$ be the user selection indicator whose value is $1$ if User~$i$ is selected in slot~$t$ and $0$ otherwise.
Then, the received signal at User~$i$ in slot~$t$ is given by
\begin{equation}\label{eq:y_i}
	y_i^t = h_i^t \sum_{i\in\calN} q_i^t \sqrt{p_i^t}\,x_i^t + n_i^t,
\end{equation}
where $n_i^t \sim \mathcal{CN} (0,\sigma_i^2)$ is the zero-mean Gaussian noise with variance~$\sigma_i^2$.
For ease description, we define the noise-to-channel ratio (NCR) of User~$i$ in slot~$t$ as
\begin{equation}
	\eta_i = \frac{\sigma_i^2}{\vert h_i^t\rvert^2}.
\end{equation}

After receiving signal $y_i^t$, User~$i$ performs SIC to decode its own signal, $x_i^t$, from it.
User~$i$ first decodes the signals for each User~$j$ with $\eta_j \ge \eta_i$, and then subtracts the components associated with them from the received signal.
In succession, User~$i$ decodes its own signal by treating the signals for users whose NCRs are smaller than its NCR as noise.
With a typical assumption that SIC can be successfully done, the maximum achievable data rate of User~$i$ in slot~$t$ is obtained as
\begin{equation}\label{eq:Rate}
	R_i(\bp^t, \bq^t; \bh^t) = q_i^t \log_2 \left( 1 + \frac{p_i^t}{\sum_{j\in\calN:\eta_j^t<\eta_i^t}p_j^t + \eta_i^t} \right),
\end{equation}
where $\bp^t = (p_i^t)_{\forall i\in\calN}$, $\bq^t = (q_i^t)_{\forall i\in\calN}$, and $\bh^t=(h_i^t)_{\forall i\in\calN}$.

\section{Opportunistic User and Power Scheduling}
\label{sec:OUPS}
We formulate a joint user and power scheduling problem, where the objective is to find optimal user selection and power allocation for each slot to maximize the weighted average sum rate while ensuring minimum average data rates of users as
\begin{IEEEeqnarray}{c'l}
	\IEEEyesnumber\IEEEyessubnumber*
	\maximize_{\bp^t,\,\bq^t,\,\forall t} & \lim_{T\to\infty} \frac{1}{T} \sum_{t=1}^{T} \sum_{i\in\calN} w_i R_i(\bp^t, \bq^t; \bh^t) \label{Prob:Obj}\\
	\subjto
	&\lim_{T\to\infty} \frac{1}{T} \sum_{t=1}^{T} R_i(\bp^t, \bq^t; \bh^t) \ge \bar{R}_i,~\forall i\in\calN,\IEEEeqnarraynumspace \label{Prob:minR}\\
	&\bp^t \in \calP, ~ \bq^t \in \calQ, ~ \forall t, \label{Prob:p,q}
\end{IEEEeqnarray}
where $w_i$ and  $\bar{R}_i$ are the weight and minimum average data rate requirement of User~$i$, respectively, $\calP=\{\bp^t\in\mathbb{R}_{\ge0}^N \mid \sum_{i\in\calN} p_i^t \le \Pmax\}$, and $\calQ=\{0,1\}^N$.
The constraints~\eqref{Prob:minR} and~\eqref{Prob:p,q} represent the minimum average data rate requirements of users and the ranges of power allocation and user selection vectors, respectively.
Note that the problem is not easy to solve due to the average over an infinite time horizon.
To address this problem, we take advantage of the fact that the long-term time average converges almost surely to the expectation for almost all realizations of the fading process by the ergodicity of the fading process.
Thereby, by denoting a channel vector in a generic slot by~$\bh$ without $t$, we can reformulate the problem~as
\begin{IEEEeqnarray}{c'c'l}
	\IEEEyesnumber\IEEEyessubnumber*
	\Problem & \maximize_{\bp^\bh,\,\bq^\bh,\,\forall \bh} & \E_{\bh} \left[ \sum_{i\in\calN} w_i R_i(\bp^{\bh}, \bq^{\bh} ; \bh) \right] \label{Prob1:Obj}\\
	&\subjto
	&\E_{\bh} \left[ R_i(\bp^{\bh}, \bq^{\bh} ; \bh) \right] \ge \bar{R}_i,~\forall i\in\calN,\IEEEeqnarraynumspace \label{Prob1:minR}\\
	&&\bp^{\bh} \in \calP, ~ \bq^{\bh} \in \calQ, ~ \forall \bh, \label{Prob1:p,q}
\end{IEEEeqnarray}
where $\bp^{\bh}=(p_i^{\bh})_{\forall i\in\calN}$ and $\bq^{\bh}=(q_i^{\bh})_{\forall i\in\calN}$.
For any slot with $\bh$, user selection and power allocation can be done according to the solution for $\bq^{\bh}$ and $\bp^{\bh}$ obtained by solving Problem~$\Problem$.

However, there is still a big challenge in solving Problem~$\Problem$.
That is, since no information on the underlying distribution of $\bh$ is provided, we need to solve the problem without such information.
To resolve the challenge, we leverage the duality and stochastic optimization theory as in~\cite{lee2016qc2linq, lee2019ehlinq}.
Accordingly, we first define a Lagrangian function, $L$, for Problem~$\Problem$ as
\begin{align}\label{eq:lagrangian}
	L(\barbp, \barbq, \blambda) &= \E_{\bh} \left[ \sum_{i\in\calN} w_i R_i(\bp^{\bh}, \bq^{\bh} ; \bh) \right] \nonumber\\
	&\quad + \sum_{i\in\calN} \lambda_i \left(\E_{\bh}\left[ R_i(\bp^{\bh}, \bq^{\bh} ; \bh) \right] - \bar{R}_i\right) \nonumber\\
	&= \E_{\bh} \left[ \sum_{i\in\calN} (w_i+\lambda_i) R_i(\bp^{\bh}, \bq^{\bh} ; \bh) \right] - \sum_{i\in\calN} \lambda_i \bar{R}_i,
\end{align}
where $\barbp = (\bp^{\bh})_{\forall\bh}$, $\barbq = (\bq^{\bh})_{\forall\bh}$, and $\blambda=(\lambda_i)_{\forall i\in\calN}$ is a nonnegative Lagrangian multiplier vector corresponding to the constraint~\eqref{Prob1:minR}.
Using~\eqref{eq:lagrangian}, we can define the dual problem as
\begin{IEEEeqnarray}{c'c'l}\label{Prob:Dual}
	\Dual & \minimize_{\blambda} & F(\blambda) \nonumber \\
	&\subjto
	&\blambda \succeq \mathbf{0}, \nonumber
\end{IEEEeqnarray}
where $\succeq$ is the elementwise inequality, $\mathbf{0}$ is a zero vector, and 
\begin{IEEEeqnarray}{rCc'l}
	\IEEEyesnumber\label{Prob:F}
	F(\blambda) & = & \maximize_{\barbp, \, \barbq} & L(\barbp, \barbq, \blambda) \\
	&&\subjto
	&\bp^{\bh} \in \calP, ~ \bq^{\bh} \in \calQ, ~ \forall \bh. \nonumber
\end{IEEEeqnarray}
Since Problem~$\Problem$ is nonconvex, there may be a duality gap even if Problem~$\Dual$ is optimally solved.
However, the duality gap vanishes in our problem, resulting in no loss of optimality.



\begin{theorem}\label{thm:zero-duality-gap}
The strong duality holds between Problem~$\Problem$ and its dual problem, Problem~$\Dual$.
\end{theorem}

\begin{proof}
See Appendix~\ref{prove:thm:zero-duality-gap}.
\end{proof}


We thus develop an algorithm that solves Problem~$\Dual$,
To this end, we first focus on finding its objective function, $F(\blambda)$.
The first term in~\eqref{eq:lagrangian} is separable for each channel vector, and the second term is independent of the decision variables, $\barbp$ and $\barbq$.
Hence, for any given Lagrangian multiplier vector, $\blambda$, the maximization in~\eqref{Prob:F} can be solved by separately solving the following subproblem for each given channel vector~$\bh$:
\begin{IEEEeqnarray}{c'c'l}\label{Prob:Dual_s}
	\Dualh & \maximize_{\bp^{\bh},\,\bq^{\bh}} & \sum_{i\in\calN} (w_i+\lambda_i) R_i(\bp^{\bh}, \bq^{\bh} ; \bh) \label{Prob:Dual_s:Obj} \nonumber\\
	&\subjto 
	&\bp^{\bh}\in\calP,\,\bq^{\bh}\in\calQ.
\end{IEEEeqnarray}
The expectation has disappeared in Problem~$\Dualh$, so that it can be solved without knowledge of the underlying distribution of the fading process once the channel realization is provided.
Thus, for given $\blambda$ and $\bh$, Problem~$\Dualh$ becomes a deterministic optimization problem for user selection and power allocation that aims to maximize the instantaneous weighted sum rate with weight $w_i+\lambda_i$ for User~$i$.
An algorithm to solve this problem, called USPA, will be developed in the next section.

We now focus back on solving Problem~$\Dual$.
Even though the optimal user selection and power allocation can be obtained for each $\bh$ and $\blambda$ by solving Problem~$\Dualh$, the underlying distribution of $\bh$ is still required to solve Problem~$\Dual$.
Nevertheless, thanks to the fact that Problem~$\Dual$ is a convex stochastic programming problem~\cite{shapiro2014lectures}, we can solve it using the stochastic subgradient method, where the Lagrangian multipliers are iteratively updated according to
\begin{equation}\label{eq:update}
	\lambda_i^{t+1} = \max\left\{0, \, \lambda_i^{t} - \zeta^{t} v_i^{t}\right\}, ~ \forall i\in\calN,
\end{equation}
where $\blambda^{t}$ and $\zeta^{t}$ are the Lagrangian multiplier vector and the step size in slot~$t$, respectively, and $\bv^t = (v_i^t)_{\forall i\in\calN}$ is the stochastic subgradient of $F(\blambda)$ with respect to $\blambda$ at $\blambda=\blambda^t$.
By Danskin's min-max theorem~\cite{bertsekas1999nonlinear}, we can determine $\bv^t$ as
\begin{equation}\label{eq:v}
	v_i^t = R_i^t - \bar{R}_i, ~ \forall i \in \calN,
\end{equation}
where $R_i^t$ is the instantaneous data rate of User~$i$ in slot~$t$, which is given by the user selection and power allocation according to the solution to Problem~$\Dualh$ with $\bh=\bh^t$~and~$\blambda=\blambda^t$.
When $\blambda$ follows this update process, it converges almost surely to the optimal solution, $\blambda^*$, to Problem~$\Dual$ if $\zeta^t$ meets~\cite{boyd2008stochastic}
\begin{equation}\label{eq:condition}
\zeta^{t} \ge 0, ~ \sum_{t=1}^{\infty} \zeta^{t} = \infty, ~ \textnormal{and} ~ \sum_{t=1}^{\infty} (\zeta^{t})^2 < \infty.
\end{equation}
The proposed OUPS is outlined in~Algorithm~\ref{Alg:Scheduling}.

\begin{algorithm}[!t]
\DontPrintSemicolon
\SetKwInOut{Input}{input}\SetKwInOut{Output}{output}
Initialize: $\blambda^{0}=\mathbf{0}$, and $t=1$\\
\For{each slot~$t$}{
	Obtain a solution to Problem~$\Dualh$ using USPA.\\
	Transmit a signal based on the obtained solution.\\
	Update $\blambda^{t}$ according to~\eqref{eq:update} and~\eqref{eq:v}, and $t \gets t+1$.
}
\caption{Opportunistic user and power scheduling}
\label{Alg:Scheduling}
\end{algorithm}

\section{User Selection and Power Allocation}
\label{sec:USPA}
In this section, we develop USPA to solve Problem~$\Dualh$.
To alleviate the difficulty of Problem~$\Dualh$ caused by the integer variables, i.e., $\bq^\bh$, we first consider a problem defined as
\begin{IEEEeqnarray}{c'c'l}
	\Qone  & \maximize_{\bp^{\bh}} & \sum_{i\in\calN} \tilde{w}_i R_i (\bp^{\bh};\bh)\nonumber \\
	&\subjto
	&\bp^{\bh}\in\calP,
\end{IEEEeqnarray}
where $\tilde{w}_i = w_i+\lambda_i$, and $R_i(\bp^\bh;\bh)$ is defined as $R_i(\bp^\bh, \bq^\bh; \bh)$ with $q_i^t=1$.
An optimal solution to Problem~$\Dualh$ can be easily obtained from that to Problem~$\Qone$.

\begin{theorem}\label{thm:user_selection_exclusion}
Let $(\bp^{\bh})^{\dagger}$ be an optimal solution to Problem~$\Qone$.
Then, the optimal solution, $\{(\bp^{\bh})^*, (\bq^{\bh})^*\}$, to Problem~$\Dualh$ can be obtained as
\begin{equation}
(\bp^{\bh})^* = (\bp^{\bh})^{\dagger} \text{ and } (\bq^{\bh})^* = (q_i^\bh)_{\forall i\in\calN},
\end{equation}
where for all $i\in\calN$, $q_i^{\bh}$ is $1$ if $(p_i^{\bh})^{\dagger}>0$ and $0$ otherwise.
\end{theorem}

\begin{proof}
See Appendix~\ref{prove:thm:user_selection_exclusion}.
\end{proof}

\noindent Hence, we focus on Problem~$\Qone$ rather than Problem~$\Dualh$.
For notational simplicity, we omit $\bh$ and assume, without loss of generality, that users are ordered such that $\eta_i>\eta_j$ if $i<j$.
Then, Problem~$\Qone$ can be equivalently reformulated as
\begin{IEEEeqnarray}{c'l'l}
	\Qtwo & \maximize_{p_i,\,\forall i\in\calN} & \sum_{i=1}^N \tilde{w}_i \log_2 \left( 1 + \frac{p_i}{\sum_{j>i} p_j + \eta_i^2} \right) \nonumber\\
	&\subjto
	&\sum_{i=1}^N p_i \le \Pmax,\nonumber\\
	&& p_i \ge 0, ~\forall i\in\calN.\nonumber
\end{IEEEeqnarray}
Problem~$\Qtwo$ is still nonconvex, so we cannot apply standard tools for convex optimization.
Furthermore, as mentioned before, the problem needs to be solved with very low computational complexity since the user selection and power allocation should be done at every slot with a very short time period.
To cope with these challenges, we first find \textit{candidate} users to whom power may be allocated, and then address how to optimally allocate power to them.
To this end, we define a \textit{last SIC user} as follows.

\begin{definition}
A \textit{last SIC user} refers to a user who does not experience any interference signal after the SIC process.
\end{definition}

We start with the assumption that User~$k$ is the last SIC user and has been allocated a certain amount of power.
Accordingly, we assume that $p_k$ is given as a certain positive value, and $p_i$'s for $i>k$ are given as zero so that User~$k$ does not experience any interference signal after the SIC process.
Under this assumption, $p_i$'s for $i\ge k$ are no longer decision variables.
Note that how to select the last SIC user and how much power to allocate to it will be discussed later.
Accordingly, Problem~$\Qtwo$ can be reformulated as
\begin{IEEEeqnarray}{c'c'l}
	\Qthree & \maximize_{p_i,\,\forall i<k} & \sum_{i=1}^{k-1} \tilde{w}_i \log_2 \left( 1 + \frac{p_i}{\sum_{j>i} p_j + \eta_i^2} \right) \nonumber\\
	&\subjto
	&\sum_{i=1}^{k-1} p_i + p_k \le \Pmax, \nonumber \\
	&& p_i \ge 0, ~\forall i<k.\nonumber
\end{IEEEeqnarray}
The purpose of this problem is not to find power allocation to users but to find \textit{candidate} users when User~$k$ is the last SIC user.
Based on the fact that the interference power is usually much greater than the noise power, the noise power of users suffering from the interference signals can be assumed to be negligible.
Hence, we assume that $\sigma_i^2=0$ for $i<k$.
Then, by letting $\rho_i = \sum_{j=i}^k p_j$, we can approximate Problem~$\Qthree$ as
\begin{IEEEeqnarray}{c'c'l}
	\Qfour & \maximize_{\rho_i,\,\forall i<k} & \sum_{i=1}^{k-1} \tilde{w}_i \log_2 \left( \frac{\rho_i}{\rho_{i+1}} \right) \nonumber\\
	&\subjto
	&\prod_{i=1}^{k-1}\frac{\rho_i}{\rho_{i+1}}\times\rho_k \le \Pmax \nonumber\\
	&& \frac{\rho_i}{\rho_{i+1}} \ge 1, ~\forall i<k, \nonumber
\end{IEEEeqnarray}
where the constraints are equivalent to those in Problem~$\Qthree$, which can be derived by simple arithmetic operations.
In succession, by letting $r_i = \log_2(\rho_i/\rho_{i+1})$ for $i<k$, $r_k=\log_2(\rho_k)$, and taking the logarithm of the both sides of the constraints, we can reformulate Problem~$\Qfour$ equivalently as
\begin{IEEEeqnarray}{c'c'l}
	\Qfive & \maximize_{r_i,\,\forall i<k} & \sum_{i=1}^{k-1} \tilde{w}_i r_i \nonumber\\
	&\subjto
	&\sum_{i=1}^{k-1} r_i + r_k \le \log_2(\Pmax), \nonumber\\
	&& r_i \ge 0, ~\forall i<k. \nonumber
\end{IEEEeqnarray}
In Problem~$\Qfive$, the decision variables, $r_i$'s for all $i<k$, are linearly combined in the objective function, and the feasible set is a unit simplex.
Hence, it is obvious that the objective function is maximized when all $r_i$'s, except the one with the largest weight, are zero.
Also, by the definition of $r_i$, we can easily see that, for any $i<k$, $p_i$ is zero if and only if $r_i$ is zero.
Thus, we can conclude that only one user with the largest weight is selected as the other \textit{candidate} user together with the last SIC user, i.e., User~$k$.
We state this result in the following theorem.
\begin{theorem}\label{thm:two user selected}
	Under the assumption that the noise signals of users who suffer from the interference signals are neglected, by the solution to Problem~$\Qtwo$, at most two users are selected as \textit{candidate} users to whom power may be allocated.
	To be specific, when User~$k$ is selected as the last SIC user, User~$\phi_k$ is accordingly selected as the other \textit{candidate} user where
	\begin{equation}\label{eq:two user selected}
		\phi_k = \argmax_{i<k} \{\tilde{w}_i\}.
	\end{equation}
\end{theorem}

\noindent By Theorem~\ref{thm:two user selected}, Problem~$\Qtwo$ can be reduced to the power allocation problem for the two-user case, defined by
\begin{IEEEeqnarray}{c'c'l}
	\Ptwo & \maximize_{p_{\phi_k}, p_k} & \tilde{w}_{\phi_k} \log_2 \left( 1 + \frac{p_{\phi_k}}{p_k + \eta_{\phi_k}^2} \right) + \tilde{w}_k \log_2 \left( 1 + \frac{p_k}{\eta_k^2} \right) \nonumber \label{def:objective}\\
	&\subjto
	&p_{\phi_k} + p_k \le \Pmax, \nonumber\\
	&&p_{\phi_k} \ge 0; ~ p_k\ge0.\nonumber
\end{IEEEeqnarray}
This two-user power allocation problem can be optimally solved in closed forms.

\begin{theorem}\label{thm:2-user}
The optimal solution, $\{p_{\phi_k}^*, p_k^*\}$, to Problem~$\Ptwo$ is derived as
\begin{align}
	&p_k^* = \begin{cases}
				0, & \textnormal{if } {\tilde{w}_k}/{\tilde{w}_{\phi_k}} < C_1, \\
				\Pmax, & \textnormal{if } {\tilde{w}_k}/{\tilde{w}_{\phi_k}} \ge C_2, \\
				\frac{\tilde{w}_{\phi_k} \eta_k^2 - \tilde{w}_k \eta_{\phi_k}^2}{(\tilde{w}_k - \tilde{w}_{\phi_k})}, & \textnormal{otherwise,} \\
		\end{cases} \label{eq:pkstar}\\
	&p_{\phi_k}^* = \Pmax - p_k^*, \label{eq:pphistar}
\end{align}
where
\begin{equation}
	C_1 = \frac{\eta_k^2}{\eta_{\phi_k}^2} \quad \textnormal{and} \quad C_2 = \frac{\Pmax+\eta_k^2}{\Pmax+\eta_{\phi_k}^2}.
\end{equation}
\end{theorem}

\begin{proof}
See Appendix~\ref{prove:thm:2-user}.
\end{proof}

\noindent Using Theorem~\ref{thm:2-user}, the weighted sum rate, $R_{\textnormal{sum}}^{k}$, when User~$k$ is selected as the last SIC user can be given as
\begin{equation}\label{eq:RSIC}
	R_{\textnormal{sum}}^{k} = \tilde{w}_{\phi_k} \log_2 \left( 1 + \frac{p_{\phi_k}^*}{p_k^* + \eta_{\phi_k}^2} \right) + \tilde{w}_k \log_2 \left( 1 + \frac{p_k^*}{\eta_k^2} \right).
\end{equation}
In turn, the optimal last SIC user, $k^*$, can be obtained as
\begin{IEEEeqnarray}{rCll}\label{eq:kstar}
	k^* & = & \argmax_{k\in\calN} R_{\textnormal{sum}}^{k}.
\end{IEEEeqnarray}
The pseudocode of USPA is described in Algorithm~\ref{Alg:PAA}.

\begin{algorithm}[!t]
\DontPrintSemicolon
\SetKwInOut{Input}{input}\SetKwInOut{Output}{output}
\For{each $k\in\calN$}{
	Suppose that User~$k$ is the last SIC user.\\
	Select the other \textit{candidate} user, $\phi_k$, using \eqref{eq:two user selected}.\\
	Calculate $R_{\textnormal{sum}}^{k}$ using \eqref{eq:RSIC}.
}
Find $k^*$ using~\eqref{eq:kstar}.\\
Select Users~$k^*$ and $\phi_{k^*}$ as the optimal \textit{candidate} users.\\
Allocate power to them according to \eqref{eq:pkstar} and \eqref{eq:pphistar}.
\caption{User Selection and Power Allocation}
\label{Alg:PAA}
\end{algorithm}

Before closing this section, we discuss the computational complexity.
Our USPA has linear computational complexity, i.e., $\mathcal{O}(N)$, since $R_{\textnormal{sum}}^{k}$ can be calculated based on the closed-form formula in~\eqref{eq:RSIC}.
As a benchmark, we consider the GP-based algorithm proposed in~\cite{di2016sub}, which provides an optimal solution to Problem~$\Qtwo$ by solving its equivalent GP problem using interior point methods.
Hence, we call it OPT.
Note that the computational complexity of OPT is known as $\mathcal{O}((k+m)^{1/2}(mk^2+k^3+n^3))$~\cite{nemirovski2004interior}, where $k$, $m$, and $n$ are the problem-dependent parameters.
In our problem, they are set as follows: $k=2N+1$, $m=N+1$, and $n=N$.
Thus, OPT has nonlinear computational complexity of $\mathcal{O}((3N+2)^{1/2}(13N^3+20N^2+11N+2) + L) \approx \mathcal{O}(13\sqrt{3}N^{7/2} + L)$, where $L$ represents the computational complexity for converting from the solution to the GP problem back to that to Problem~$\Qtwo$.
%
As a result, the computational complexity of OPT is much higher than that of USPA, and thus it is a heavy burden to run OPT at every slot.
We will show the comparison between USPA and OPT in detail in the next section.

\section{Simulation Results}
\label{sec:sim}
We consider a single cell with one BS with a maximum transmission power of \SI{43}{\dBm} and \num{5} users.
According to~\cite{3gpp.36.814}, we set the large-scale path loss to $128.1+37.6\log_{10}(d_{\textnormal{km}})$\,\si{\dB}, where $d_{\textnormal{km}}$ is the distance in kilometers, and consider the lognormal shadow fading with a standard deviation of \SI{8}{\dB} and the small-scale fading with coefficients following independent and identical zero-mean unit-variance complex Gaussian distributions.
The noise power for each user is set to \SI[per-mode=symbol]{-104}{\dBm}.
The step size of the stochastic subgradient algorithm in~\eqref{eq:update} is set to $\zeta^t = 1/t$, which satisfies the conditions in~\eqref{eq:condition} so that the convergence of the algorithm is guaranteed.

We first compare USPA and OPT in solving Problem~$\Qtwo$.
The comparison results for \num{1000} independent trials are shown in Fig.~\ref{fig:USPA vs OPT}.
In each trial, user weights are randomly set between \num{0} and \num{1} and then normalized by their sum, and the distance of each user from the BS is randomly set between \SI{20}{\meter} and \SI{500}{\meter}.
Fig.~\ref{fig:USPA vs OPT:wsr} shows the performance comparison results over \num{1000} trials in terms of the weighted sum rate.
As shown in the figure, the performance of USPA is very close to that of OPT.
On average, the performance difference between the two methods over \num{1000} trials is only \SI[per-mode=symbol]{0.0412}{\bps\per\Hz} (\SI{0.7}{\percent}).
Fig.~\ref{fig:USPA vs OPT:count} shows the frequency histogram of the number of selected users, and Fig.~\ref{fig:USPA vs OPT:awdr} shows the weighted average data rates of users, where the user indices are sorted in decreasing order of the weighted data rate.
According to Fig.~\ref{fig:USPA vs OPT:count}, in OPT, the probability of selecting three or more users reaches about \SI{25}{\percent}, which is not small.
However, according to Fig.~\ref{fig:USPA vs OPT:awdr}, since the sum of the weighted data rates of the first two users occupies over \SI{96}{\percent} of the weighted sum rate, the last three users do not significantly contribute to the performance.
That is why the performance of USPA, which selects at most two users, is very close to that of OPT.
In addition, it is worth noting that in obtaining the above simulation results, OPT have taken about \num{3900}~times more execution time than USPA.\footnote{All simulation results have been obtained by using MATLAB R2019b on a computer with Intel Core i7-9700K CPU (3.60 GHz) and 32.0 GB RAM.}
These results verify that not only does USPA provide good performance close to the optimal one, but it also has very low computational complexity.

\begin{figure}[!t]
\centering
\subfloat[Weighted sum rate]{\includegraphics[width=.94\linewidth]{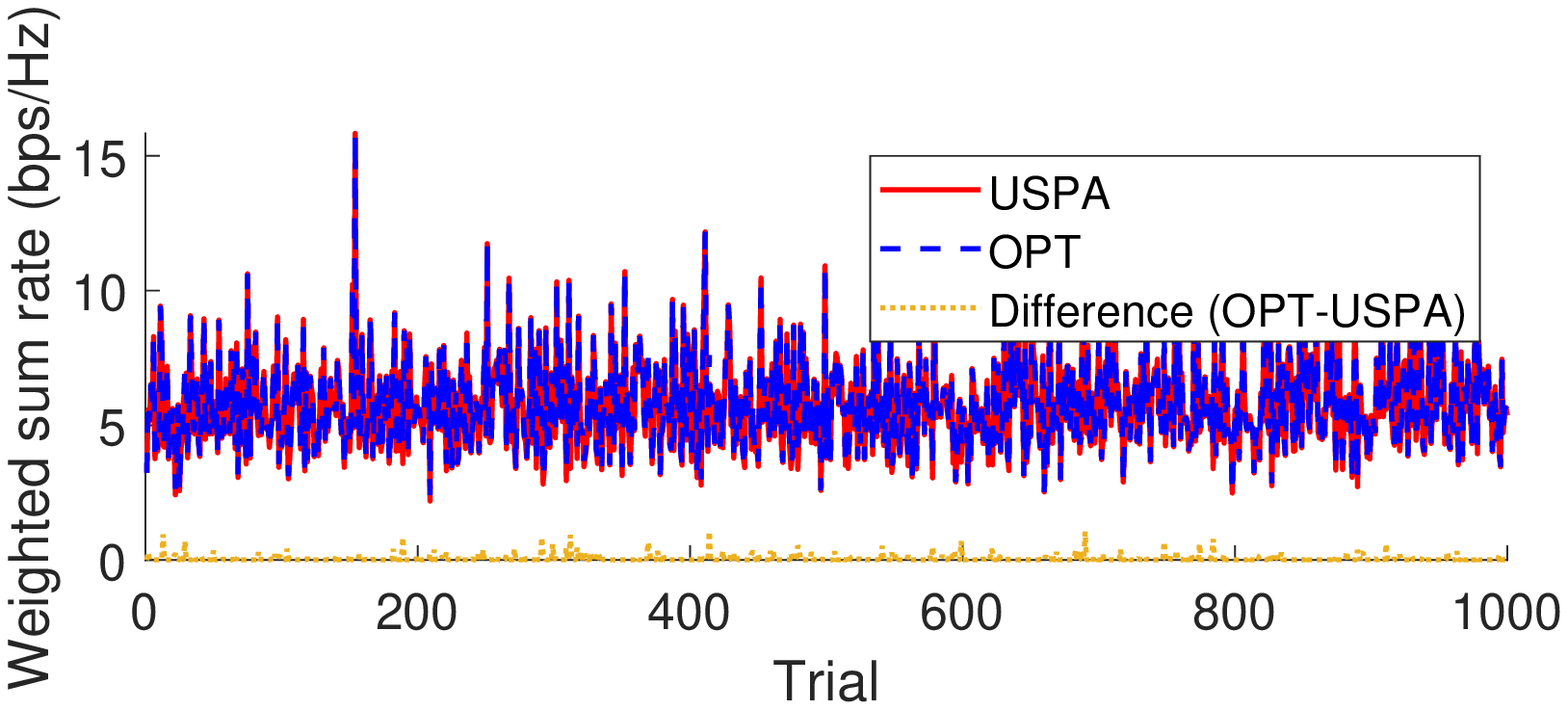}\label{fig:USPA vs OPT:wsr}}\\
\subfloat[Frequency of the number of selected users]{~\includegraphics[width=.47\linewidth]{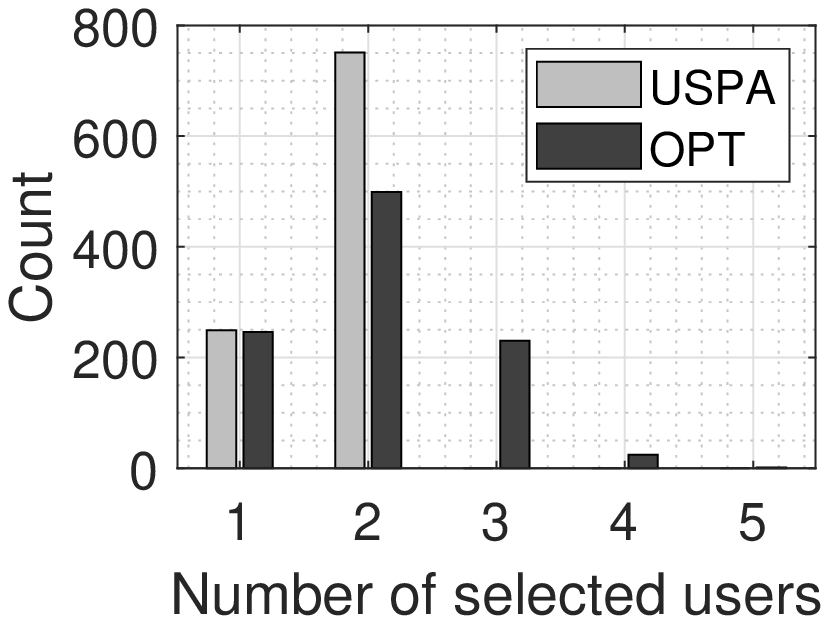}~\label{fig:USPA vs OPT:count}}
\hfil
\subfloat[Weighted data rate for each user]{~\includegraphics[width=.47\linewidth]{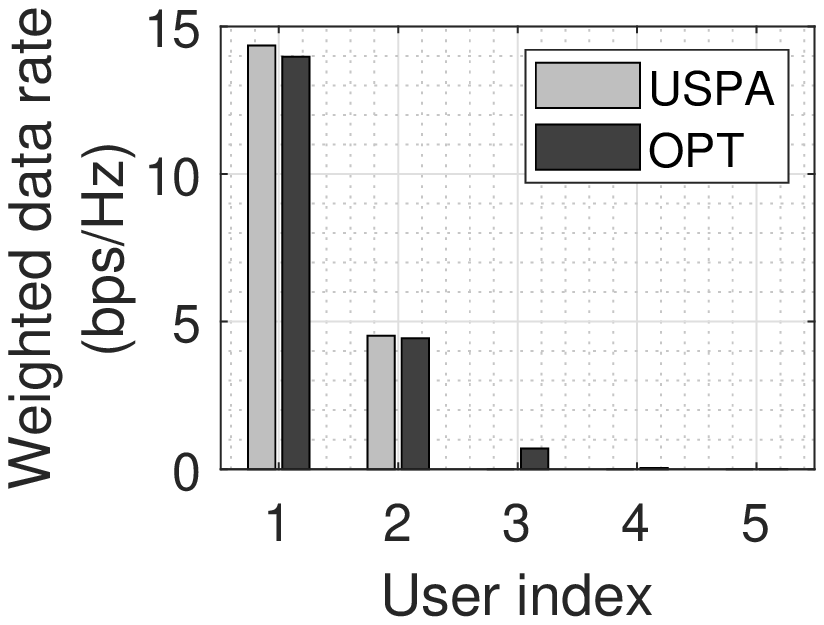}~\label{fig:USPA vs OPT:awdr}}
\caption{Comparison results between USPA and OPT.\label{fig:USPA vs OPT}}
\end{figure}

Now, we evaluate the performance of our OUPS by comparing the following three types of OUPS:
(i) OUPS-USPA where user selection and power allocation are performed by our USPA, (ii) OUPS-OPT where they are performed by OPT, and (iii) OUPS-OMA where only one user who can provide the highest instantaneous weighted data rate using full power is selected at each slot.
Fig.~\ref{fig:OUPS} shows the performance results over \num[group-separator={,}]{10000} slots under the scenario where \num{5} users with equal weights are located \SI{20}{\meter}, \SI{140}{\meter}, \SI{260}{\meter}, \SI{380}{\meter}, and \SI{500}{\meter} away from the BS.
The first three users have the minimum average data rate requirements of \SI[per-mode=symbol]{2}{\bps\per\Hz}, and the last two users have those of \SI[per-mode=symbol]{4}{\bps\per\Hz}.
Fig.~\ref{fig:OUPS:sum rate} shows that OUPS-USPA not only provides much higher performance than OUPS-OMA, but also provides high performance comparable to OUPS-OPT.
On the other hand, Fig.~\ref{fig:OUPS:data rate} shows the average data rates of each user for our OUPS-USPA.
As shown in the figure, we can see that the QoS constraints for all users are well satisfied, and the user closest to the BS achieves the highest average data rate among all users so that the average sum rate is maximized.
The results demonstrate that OUPS-USPA provides near-optimal performance while ensuring the given QoS constraints.

\begin{figure}[!t]
\centering
\subfloat[Average sum rate $(\bar{R}_i = 2, \forall i)$]{\includegraphics[width=.94\linewidth]{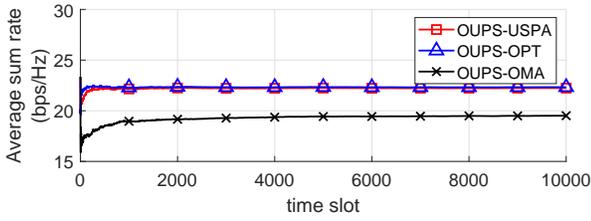}\label{fig:OUPS:sum rate}}
\hfil
\subfloat[Average data rate of each user]{\includegraphics[width=.94\linewidth]{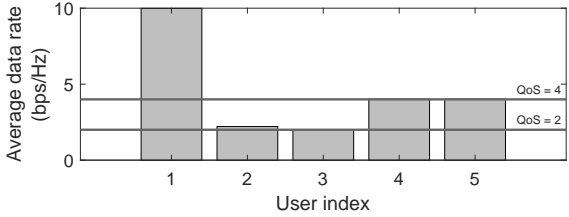}\label{fig:OUPS:data rate}}
\caption{Performance results of OUPS.\label{fig:OUPS}}
\end{figure}

\section{Conclusion}
\label{sec:conc}
We have studied the joint user and power scheduling problem to maximize the weighted average sum rate while ensuring given QoS constraints.
We have first developed OUPS that fully exploits time-varying channels, and then developed USPA with extremely low computational complexity.
Through the simulation results, we have shown that USPA provides near-optimal performance despite being about \num{3900}~times faster than OPT, and OUPS meets the QoS constraints well.
This study will be the cornerstone of our future work on scheduling for multi-carrier NOMA systems.

\appendices
\numberwithin{equation}{section}
\numberwithin{definition}{section}
\numberwithin{lemma}{section}
\section{Proof of Theorem~\ref{thm:zero-duality-gap}}
\label{prove:thm:zero-duality-gap}

In order to prove the strong duality between Problem~$\Problem$ and its dual problem, Problem~$\Dual$, we utilize the \textit{time-sharing} condition proposed in~\cite{yu2006dual}, which is defined as follows.
\begin{definition}
Let $\{\barbp_x, \barbq_x\}$ and $\{\barbp_y, \barbq_y\}$ be the optimal solutions to Problem~$\Problem$ with $\bar{\bR}=\bar{\bR}_x$ and $\bar{\bR}=\bar{\bR}_y$, respectively, where $\bar{\bR}=(\bar{R}_i)_{\forall i\in\calN}$, $\bar{\bR}_x=(\bar{R}_{x,i})_{\forall i\in\calN}$, and $\bar{\bR}=(\bar{R}_{y,i})_{\forall i\in\calN}$.
Then, Problem~$\Problem$ is said to satisfy the time-sharing condition if for any $\bar{\bR}_x$ and $\bar{\bR}_y$, and for any $0\le\theta\le1$, there always exists a feasible solution $\{\barbp_z, \barbq_z\}$ such that
\begin{equation}\label{eq:time-sharing-condition-1}
	\E_{\bh} \left[R_i(\bp_z^{\bh},\bq_z^\bh;\bh)\right] \ge \theta \bar{R}_{x,i} + (1-\theta) \bar{R}_{y,i}, ~\forall i\in\calN,
\end{equation}
\begin{multline}\label{eq:time-sharing-condition-2}
	\E_{\bh} \left[ \sum_{i\in\calN} w_i R_i(\bp_z^{\bh}, \bq_z^\bh ; \bh) \right]
	\ge \theta \E_{\bh} \left[ \sum_{i\in\calN} w_i R_i(\bp_x^{\bh},\bq_x^\bh ; \bh) \right]\\
	+ (1-\theta) \E_{\bh} \left[ \sum_{i\in\calN} w_i R_i(\bp_y^{\bh},\bq_y^\bh ; \bh) \right].
\end{multline}
\end{definition}

It has been proven in~\cite{yu2006dual} that if an optimization problem satisfies the time-sharing condition, the strong duality holds regardless of the convexity of the problem.
Hence, we prove Theorem~\ref{thm:zero-duality-gap} by showing that Problem~$\Problem$ satisfies the time-sharing condition.
First, for any $\{\barbp_x,\barbq_x\}$ and $\{\barbp_y,\barbq_y\}$, and for any $\theta\in[0,1]$, let us set $\{\bp_z,\bq_z\}$ as
\begin{equation}\label{eq:appendixA:feasible_solution}
	\{\bp_z^t, \bq_z^t\} = \begin{cases} \{\bp_x^t, \bq_x^t\} & t\le \lfloor\theta T\rfloor,\\ \{\bp_y^t, \bq_y^t\}, & t \ge \lfloor\theta T+1\rfloor, \end{cases}
\end{equation}
where $\lfloor\cdot\rfloor$ is the floor function that gives the largest integer not exceeding its argument.
Then, the first condition~\eqref{eq:time-sharing-condition-1} holds as follows.
For all $i\in\calN$,
\begin{IEEEeqnarray}{rCl}
	\IEEEeqnarraymulticol{3}{l}{%
		\E_{\bh} \left[ R_i(\bp_z^{\bh}, \bq_z^{\bh} ; \bh) \right]
	}\nonumber\\
	& = & \lim_{T\to\infty} \frac{1}{T} \sum_{t=1}^{T} R_i(\bp_z^t, \bq_z^t; \bh^t) \nonumber\\
	& = & \lim_{T\to\infty} \frac{1}{T} \Biggl( \sum_{t=1}^{\lfloor \theta T \rfloor} R_i(\bp_x^t, \bq_x^t; \bh^t) + \sum_{t=\lfloor\theta T+1\rfloor}^{T} R_i(\bp_y^t, \bq_y^t; \bh^t) \Biggr) \nonumber\\
	& = & \theta \E_{\bh} \left[ R_i(\bp_x^{\bh}, \bq_x^\bh; \bh) \right] + (1-\theta) \E_{\bh} \left[ R_i(\bp_y^{\bh}, \bq_y^\bh; \bh) \right] \nonumber\\
	& \ge & \theta \bar{R}_{x,i} + (1-\theta) \bar{R}_{y,i}.
\end{IEEEeqnarray}
where the first and third equality holds due to the ergodicity of the fading process.
Similarly, the second condition~\eqref{eq:time-sharing-condition-2} also holds as follows.
\begin{IEEEeqnarray}{rCl}
	\IEEEeqnarraymulticol{3}{l}{%
		\E_{\bh} \left[ \sum_{i\in\calN} w_i R_i(\bp_z^{\bh}, \bq_z^\bh ; \bh) \right]
	}\nonumber\\
	& = & \lim_{T\to\infty} \frac{1}{T} \sum_{t=1}^{T} \sum_{i\in\calN} w_i R_i(\bp_z^t, \bq_z^t; \bh^t) \nonumber\\
	& = & \lim_{T\to\infty} \frac{1}{T} \Biggl( \sum_{t=1}^{\lfloor \theta T \rfloor} \sum_{i\in\calN} w_i R_i(\bp_x^t, \bq_x^t; \bh^t) \nonumber\\
	&& \qquad\qquad\qquad +\> \sum_{t=\lfloor\theta T+1\rfloor}^{T} \sum_{i\in\calN} w_i R_i(\bp_y^t, \bq_y^t; \bh^t) \Biggr) \nonumber\\
	& = & \theta \E_{\bh} \left[ \sum_{i\in\calN} w_i R_i(\bp_x^{\bh}, \bq_x^\bh ; \bh) \right] \nonumber\\
	&& \qquad\qquad\qquad +\> (1-\theta) \E_{\bh} \left[ \sum_{i\in\calN} w_i R_i(\bp_y^{\bh}, \bq_y^\bh; \bh) \right]. \IEEEeqnarraynumspace
\end{IEEEeqnarray}
Hence, the time-sharing condition holds for Problem~$\Problem$. \hfill \qedsymbol

\section{Proof of Theorem~\ref{thm:user_selection_exclusion}}
\label{prove:thm:user_selection_exclusion}
Since $q_i^\bh$ is either $0$ or $1$, the objective value of Problem~$\Dualh$ is less than or equal to the optimal value of Problem~$\Qone$, i.e., $\sum_{i\in\calN}\tilde{w}_i R_i(\bp^\bh, \bq^\bh; \bh) \le \sum_{i\in\calN} \tilde{w}_i R_i(\bp^\bh; \bh)$, for any $\bp^\bh\in\calP$ and $\bq^\bh\in\calQ$.
Hence, by letting $(\bp^\bh)^\dagger$ be an optimal solution to Problem~$\Qone$, we have
\begin{equation}
	\sum_{i\in\calN}\tilde{w}_i R_i(\bp^\bh, \bq^\bh; \bh) \le \sum_{i\in\calN} \tilde{w}_i R_i((\bp^\bh)^\dagger; \bh),
\end{equation}
for all $\bp^\bh\in\calP$ and $\bq^\bh\in\calQ$.
%
%
Consider $(\bp^\bh)^*=(\bp^\bh)^\dagger$ and $(\bq^\bh)^*=(q_i^\bh)_{i\in\calN}$ such that $q_i^\bh=1$ if $(p_i^\bh)^\dagger>0$ and $q_i^\bh=0$ otherwise.
By simple arithmetic operations, it can be easily verified that
\begin{equation}
	\sum_{i\in\calN}\tilde{w}_i R_i((\bp^\bh)^*, (\bq^\bh)^*; \bh) = \sum_{i\in\calN} \tilde{w}_i R_i((\bp^\bh)^\dagger; \bh).
\end{equation}
Hence, an optimal solution to Problem~$\Dualh$ is given as $\{(\bp^\bh)^*, (\bq^\bh)^*\}$, which is obtained from the optimal solution, $(\bp^\bh)^\dagger$, to Problem~$\Qone$.
\hfill \qedsymbol

\section{Proof of Theorem~\ref{thm:2-user}}
\label{prove:thm:2-user}
Since a larger transmission power results in a higher weighted sum rate, the first constraint in Problem~$\Ptwo$ can be replaced with $p_{\phi_k} + p_k = \Pmax$.
Then, by substituting $p_{\phi_k}$ into $\Pmax-p_k$, Problem~$\Ptwo$ can be equivalently transformed into a one-variable optimization problem as
\begin{IEEEeqnarray}{c'c'l} \label{prob:P2_prime}
	\Ptwop & \maximize_{0 \le p_k \le \Pmax} & g(p_k),
\end{IEEEeqnarray}
where
\begin{equation}
g(p_k) = \tilde{w}_{\phi_k} \log_2 \left( \frac{\Pmax + \eta_{\phi_k}^2}{p_k + \eta_{\phi_k}^2} \right) + \tilde{w}_k \log_2 \left(\frac{p_k + \eta_k^2}{\eta_k^2} \right).
\end{equation}
The derivative of $g(p_k)$ with respect to $p_k$ is given as
\begin{align}\label{eq:gprime}
	g'(p_k) &= \frac{1}{\ln2} \left[ \frac{-\tilde{w}_{\phi_k}}{p_k + \eta_{\phi_k}^2} + \frac{\tilde{w}_k}{p_k + \eta_k^2} \right] \nonumber\\
	&= \frac{1}{\ln2} \cdot \frac{(\tilde{w}_k - \tilde{w}_{\phi_k}) p_k + \tilde{w}_k \eta_{\phi_k}^2 - \tilde{w}_{\phi_k} \eta_k^2}{(p_k + \eta_{\phi_k}^2)(p_k + \eta_k^2)}.
\end{align}
From the above equation, we have $g'(\hat{p}_k) = 0$ if and only if
\begin{equation}\label{eq:p2hat}
	\hat{p}_k = \frac{\tilde{w}_{\phi_k} \eta_k^2 - \tilde{w}_k \eta_{\phi_k}^2}{(\tilde{w}_k - \tilde{w}_{\phi_k})}.
\end{equation}
Using~\eqref{eq:gprime} and \eqref{eq:p2hat}, we can derive an optimal solution, $p_k^*$, to the problem in~\eqref{prob:P2_prime} by considering the following three mutually exclusive cases:
\begin{enumerate}
\item 
		Suppose $\tilde{w}_k/\tilde{w}_{\phi_k} > 1$.
		Then, according to \eqref{eq:gprime}, $g'(p_k) < 0$ if $p_k < \hat{p}_k$ and $g'(p_k) \ge 0$ otherwise.
		Also, according to \eqref{eq:p2hat}, $\hat{p}_k < 0$ since $\eta_k<\eta_{\phi_k}$.
		Thus, $g(p_k)$ is an increasing function on $[0, \Pmax]$, resulting in $p_k^*=\Pmax$.
\item
		Suppose $\tilde{w}_k/\tilde{w}_{\phi_k} = 1$.
		Then, according to \eqref{eq:gprime}, $g'(p_k) \ge 0$ for any $p_k \in [0, \Pmax]$.
		Thus, $g(p_k)$ is an increasing function on $[0, \Pmax]$, resulting in $p_k^*=\Pmax$.
\item
		Suppose $\tilde{w}_k/\tilde{w}_{\phi_k} < 1$.
		Then, according to \eqref{eq:gprime}, $g'(p_k)>0$ if $p_k<\hat{p}_k$ and $g'(p_k)\le0$ otherwise.
		Also, according to \eqref{eq:p2hat}, we have the following two inequalities.
		\begin{align}
			\hat{p}_k < 0 ~ &\Leftrightarrow ~ \frac{\tilde{w}_k}{\tilde{w}_{\phi_k}} < \frac{\eta_k^2}{\eta_{\phi_k}^2}, \label{eq:ineq:C5}\\
			\hat{p}_k \ge \Pmax ~ &\Leftrightarrow ~ \frac{\tilde{w}_k}{\tilde{w}_{\phi_k}} \ge \frac{\Pmax+\eta_k^2}{\Pmax+\eta_{\phi_k}^2}, \label{eq:ineq:C6}
		\end{align}
		Consequently, if \eqref{eq:ineq:C5} is met, $g(p_k)$ is a decreasing function on $[0,\Pmax]$, resulting in $p_k^*=0$; if \eqref{eq:ineq:C6} is met, $g(p_k)$ is an increasing function on $[0,\Pmax]$, resulting in $p_k^*=\Pmax$; and if neither \eqref{eq:ineq:C5} nor \eqref{eq:ineq:C6} is met, $g(p_k)$ is an increasing function on $[0,\hat{p}_k]$ but a decreasing function on $[\hat{p}_k,\Pmax]$, resulting in $p_k^*=\hat{p}_k$.
%
\end{enumerate}
Note that since $\eta_k < \eta_{\phi_k}$, the right-hand side of \eqref{eq:ineq:C6} is always less than one, i.e., $(\Pmax+\eta_k^2)/(\Pmax+\eta_{\phi_k}^2)<1$.
Hence, by combining the results for the above three cases, we can derive $p_k^*$ as in \eqref{eq:pkstar}.
Lastly, we can derive $p_{\phi_k}^*$ as in \eqref{eq:pphistar} since $p_{\phi_k}+p_k=\Pmax$. \hfill \qedsymbol
%
%
%

\bibliographystyle{IEEEtran}
\bibliography{IEEEabrv,vtc2021}


\end{document}